\documentclass[12pt]{article}
\usepackage{amsfonts}
\usepackage{amssymb}
\usepackage{graphicx}
\usepackage{amsmath}
\usepackage{color}
\usepackage{ulem}
\usepackage{subcaption}

\newtheorem{theorem}{Theorem} [section]

\newtheorem{definition}[theorem]{Definition}
\newtheorem{example}[theorem]{Example}

\newtheorem{lemma}[theorem]{Lemma}

\newtheorem{proposition}[theorem]{Proposition}

\newenvironment{proof}[1][Proof]{\textbf{#1.} }{\ \rule{0.5em}{0.5em}}
\DeclareCaptionLabelSeparator{custom}{}
\allowdisplaybreaks[1]

\definecolor{armygreen}{rgb}{0.19, 0.53, 0.43}
\definecolor{atomictangerine}{rgb}{1.0, 0.6, 0.4}

\begin{document}

\author{Luis A. Guardiola\thanks{Departamento de Métodos Cuantitativos para la Economía y la Empresa, Universidad de Murcia, Murcia 30100, Spain. e-mail:%
 guardiola@um.es} $^{,}$\thanks{Corresponding author.}, Ana Meca\thanks{I. U. Centro de Investigación Operativa, Universidad Miguel Hernández de Elche, 03202 Elche, Spain. e-mail: ana.meca@umh.es} and Justo Puerto\thanks{IMUS, Universidad de Sevilla, 41012 Sevilla, Spain. e-mail: puerto@us.es}}
\title{Allocating the surplus induced by cooperation in distribution chains
with multiple suppliers and retailers \thanks{
We gratefully acknowledge financial support from the Ministerio de Ciencia,
Innovaci\'on y Universidades (MCIU/AEI/FEDER, UE) through project PGC2018-097965-B-I00, PID2020-114594GB--\{C21,C22\}; P18-FR-1422;  and Fundación BBVA: project NetmeetData (Ayudas Fundación BBVA a equipos de investigación científica 2019);
and by Generalitat Valenciana through project PROMETEO/2021/063 .}}
\maketitle

\begin{abstract}
The coordination of actions and the allocation of profit in supply chains
under decentralized control play an important role in improving the profits
of retailers and suppliers in the chain. We focus on supply chains under
decentralized control in which noncompeting retailers can order from multiple suppliers to replenish their stocks. Suppliers' production capacity is bounded. The goal of the firms in the chain is to maximize their individual profits. As the outcome under decentralized control is inefficient, coordination of actions between cooperating agents can improve individual profits. Cooperative game theory is used to analyze cooperation between agents. We define multi-retailer-supplier games and show that agents can always achieve together an optimal profit and they have incentives to cooperate and to form the grand coalition. Moreover, we show that there always exist stable allocations of the total profit among the firms upon which no coalition can improve. Then we propose and characterize a stable allocation of the total surplus induced by cooperation.

\bigskip \noindent \textbf{Key words:} decentralized distribution chains,
multiple suppliers, cooperation, stable allocations

\noindent \textbf{2000 AMS Subject classification:} 91A12, 90B99
\end{abstract}

\newpage

\section{Introduction}

The coordination of actions and the allocation of profit in supply chains
under decentralized control play a important role in improving the profits
of retailers and suppliers in the chain. The study of the coordination of actions
and the allocation of profits in distribution chains has already been
studied using cooperative game theory.

The literature on profit allocation of supply chain models is rich and focusses both on theoretical and practical issues. Motivated by the content of this paper, we will restrict ourselves to revise the main results that contribute to the issue of profit allocation analyzed from the point of view of cooperative game theory. As earlier as in the seventies Iyer and Bergen (1997) described the issue of the distribution of benefits among partners pointing out that in many cases the outcome of cooperation was that retailer got profit whereas the manufacturers got nothing. Gavirneni et al. (1999) dealt with a production and profit distribution problem of a one-to-many secondary supply chains analyzing what affected the stability and growth of those models.  Hartman et al. (2000) studied the core of inventory centralization games. 
Ball (2001) applied the Shapley value to a cooperative game of portfolio evaluation, analyzing the rationality of various profit allocation methods. Müller et al. (2002) proved that the core of the newsvendor game was not empty. Granot and Yin (2004) analyzed competition and
cooperation aspects on decentralized push and pull assembly
systems. Timmer et al. (2004), based on probabilistic rules, proposed three Shapley-like rules of allocation for these situations. Slikker et al. (2005) extended the analysis of newsvendor games to the case with transhipments. Perea et al. (2009) also analyze profit allocation on supply chains extending Owen type solutions for these models. 
 Nagarajan and Sošić (2009) study the stability and coalition formation in
assembly models. Also  Marcotte et al. (2009) and  Ming et al. (2014) study some alternative models of cooperation for supply chain management. More recently, Moradinasab et al. (2018) revisit 
game theoretic aspects of competition and cooperation between supply chains in multi-objective petroleum green supply chain. Chen et al. (2019) make an application of cooperative game theory to the analysis of technological spillover, power and coordination of firms’ green  behaviour in a supply chain. Finally, Halat et al. (2021) study carbon tax policy in inventory games of multi-echelon supply chains.

The interested reader is also referred to the surveys (Drexl and Kimms, 1997; Leng and Parlar, 2005; Nagarajan and Sošić, 2008; Meca and Timmer, 2008; Elomri, 2015; Deng and Li, 2017), and the references therein for further details and references on the issue of cooperation in supply chain models.

Cooperation and profit allocation in single-period distribution
chains consisting of a supplier and multiple noncompeting retailers were studied in Guardiola et al. (2007). They analyzed cooperation in these chains through their corresponding cooperative TU-games, the RS-games. They showed that there is always a
coalitionally stable profit allocation for RS-games; i.e., the core of RS-games is nonempty. In addition, any of these stable allocations (core-allocations) can naturally be interpreted in terms of its subjacent supply chain. An outstanding core-allocation for RS-games is the minimal gain per capita solution
(mgpc-solution), whose characterization is provided. This
allocation is suitable and compensates the supplier for its role in achieving full cooperation.  
In this paper, we further extend the model of Guardiola et al. (2007) with a
single supplier to the more realistic case of multiple suppliers. We consider distribution chains with a single period and single product. In this supply chain, retailers place this orders at the suppliers one-time. Once the suppliers have completed the production process, they deliver the products to the retailers via a warehouse, which acts as an intermediary at no cost.
Retailers sell the products in their own separate markets without competing with each other. The retailer's expected unit revenue decreases the larger the order quantity. Each retailer seeks to maximize profit by choosing the optimal order quantity.

Retailers purchase the product at a wholesale unit price. This price is given by a function that decreases with increasing quantity ordered. Hence, retailers have incentives to cooperate with each other by combining their orders into one large order, allowing them to enjoy a lower wholesale price per unit. This practice is allowed because warehouses inform suppliers of the size of the order but without identifying the retailers placing the order. In addition, retailers could cooperate with suppliers, this causes intermediate wholesale prices to disappear and leads to a further reduction of cost inefficiencies. Obviously, the total profit of any group of agents is higher than the sum of the individual profits.

In these supply chains cooperation makes sense, so we use cooperative game theory to study them. We define a cooperative TU-game associated with this type of situation where retailers and suppliers are the players and we call it multi-retailer-supplier game (MRS-game). The value of a coalition of retailers and a coalition of suppliers equal to the optimal profit they can
achieve together. We prove that they can always achieve together an optimal
profit (the value of any coalition of retailers and suppliers is positive) and that the profit of the grand coalition of suppliers and retailers
is always greater than the sum of the profits of any partition of the sets
of retailers and suppliers (i.e., MRS-games are superadditive).
Consequently, retailers and suppliers have incentives to cooperate and to form the grand
coalition in MRS-games. We show that the core of
MRS-games is never empty, that is, all
firms in the chain are willing to cooperate because there exist stable
allocations of the total profit among the firms upon which no coalition can
improve.

We introduce a specific allocation of the total surplus induced by full cooperation, the so called Supplier Compensation allocation (SC-allocation). This allocation always belongs to the core of the MRS-game, i.e. it is a stable distribution of profits, and it satisfies several nice properties. In addition, a characterization of the SC-allocation is provided. Finally, we focus on a particular class of MRS-games with unbounded production capacity. That is, suppliers can produce sufficiently large amounts of product so that retailers are not limited in their orders. 
We propose as an allocation for this subclass of MRS-games a modification of the SC-allocation that takes into account the importance of optimal suppliers (i.e., those that
can achieve maximum profit on their own) to achieve full cooperation.

Summarizing, our contributions extend the results in Guardiola et al. (2007) in two important realistic respects:
\begin{itemize}
\item From the case of a single supplier to multiple suppliers.
\item From a general situation without production capacity restrictions to the case with bounded production capacity. We also prove balancedness for the general case and provide a new allocation rule that satisfies insightful properties.
\end{itemize}

The paper is organized as follows. We begin with a preliminary section in
cooperative game theory and other issues. Next, in Section 3, we describe the situation with
a retailer and multiple suppliers. In Section 4, we focus on cooperation in situations with multiple retailers and suppliers. We introduce and analyze the class of multi-retailer-suppplier games (henceforth MRS-games). We prove that the altruistic allocation, which shares all the profit between retailers and suppliers who do not make any profit, is stable in the sense of the core. Then, Section 5 focusses on the study of a new allocation rule for MRS-games that compensates the suppliers by reducing the retailers' profit, that is the SC-allocation. We then characterized it with desirable properties. In Section 6, we analyze a particular class of MRS-games with unbounded production capacity and propose a modification of the SC-allocation that takes into account the importance of optimal suppliers. Finally, Section 7 draws conclusions and points out further research for scholars and practitioners in the field.

\section{Preliminaries}

For the sake of readability, we include in this section most of the basic concepts in cooperative game theory that will be necessary to understand and check the details of the results in the paper. 

A cooperative (profit) TU-game is a pair $(N,v)$ where $%
N=\{1,2,...,n\}$ is a finite set of players. Let $\mathcal{P}(N)$ be the set
of all coalitions $S$ in $N$ and $v:\mathcal{P}(N)\longrightarrow 
\mathbb{R}
$ the characteristic funcion satisfying that $v(\emptyset )=0.$ Note that $v(S)$ represents the maximum profit obtained by coalition $S \subseteq  N$, and $N$ is usually called the grand coalition. A profit vector or allocation is denoted by $x\in 
\mathbb{R}
^{\left\vert N\right\vert }$, where $\left\vert N\right\vert$ is the cardinal of the grand coalition. 

A TU-game $(N,v)$ is said to be monotone increasing  if the larger coalitions obtain greater benefits, i.e. $v(S)\leq v(T)$ for all coalitions $S\subseteq T\subseteq N.$ Additionally, we say that $(N,v)$ is superadditive if any two disjoint coalitions join together, the benefit obtained is at least the sum of their separate benefits. That is, $v(S\cup T) \geq v(S)+v(T),$ for all disjoint coalitions $S,T\subseteq N.$ Note that in superadditive games it makes sense for the grand coalition to form.  Indeed, the benefit obtained by the grand coalition is at least the sum of the benefits of any other coalition and its complementary, i.e. $v(N) \geq v(S)+v(N \setminus S),$ for all $S \subseteq N.$

The core of the game $(N,v)$ is denoted by $Core(N,v)$ and consists
of all those vectors that allocate the benefit of the grand coalition efficiently and are also coalitionally stable, i.e., no player has
an incentive to leave the grand coalition, so that each coalition receives at least its
profit given by the characteristic function: $$Core(N,v)=\left\{ x\in 
\mathbb{R} ^{\left\vert N\right\vert }:x(N)=v(N)\text{ and }\sum_{i\in N}x_{i}\geq
v(S)\ \text{\ for all }S\subset N\right\} .$$ 
A TU-game is balanced if and only if the core is nonempty (Bondareva, 1963; Shapley, 1967). 

A single-valued solution $\varphi $ is an application that assigns to each TU game $(N,v)$ an allocation of $v(N)$, the profit obtained by the gran coalition. Formally, 
$\varphi:G^{N}\longrightarrow \mathbb{R}^{\left\vert N\right\vert }$, where $G^{N}$ is the class of all TU-games with player set $N\ $and $\varphi _{i}(v)$ is the profit assigned to player $i\in N$ in the game $v\in G^{N}.$ Hence,  $\varphi(v)=(\varphi _{i}(v))_{i \in N}$ is a profit vector or allocation of $v(N)$. For a complete description of cooperative game theory, we refer the reader to  Gonz\'alez-D\'{\i}az et al. (2010).

Finally, for any $a,b\in \mathbb{R}$ with $a< b,$ we denote by $[a,b]$ and $(a,b)$ a closed interval and an open interval in $\mathbb{R}$, respectively. Let $f:\mathbb{R}\longrightarrow \mathbb{R}$ a real function. It is said to be non increasing on  an real interval $I\subset \mathbb{R}$ if for any $x,y \in I$ such that $x <y$, we have that $f(x) \geq f(y)$. Finally, $\mathbb{R}_{+}^{\left\vert N\right\vert }$ is the set containing the $\left\vert
N\right\vert $-dimensional vectors with non-negative components.

\section{Problems with a retailer and multiple suppliers}

In this section we focus on a single-product supply chain problem with a single period, where a single retailer demand's is replenished by several suppliers via a warehouse. Denote by $M=\{1,...,m\}$ the set of suppliers.

In order to analyze these situations we assume the following natural
assumptions on the essential elements of the model.

\begin{itemize}
\item Production cost function. For any retailer $j\in M$, $c_{j}:[0,+\infty
)\rightarrow (0,+\infty )$\ is a decreasing and continuous function where $%
c_{j}(q)q$ is non decreasing. It is the unit production costs of the
supplier $j\in M$.

\item Wholesale price function. For any retailer $j\in M$, the wholesale function $%
w_{j}:[0,+\infty )\rightarrow (0,+\infty )$\ satisfies that $%
w_{j}(q)>c_{j}(q)$\ for all $q\geq 0$ and it is a non increasing and
continuous function. Moreover, $w_{j}(q)q$ is non decreasing. %
This function models the unitary wholesale price of the supplier.

\item Retailer price function. $p:[0,+\infty )\rightarrow \mathbb{R}$\ is
non increasing and continuous function on $q$ and satisfies $p(0)>w_{j}(0)$\
for all $j\in M$. Moreover, for each supplier $j\in M$ there exists a
positive quantity $q^{\ast }$\ such that $p(q^{\ast })=0$. This function
models the unitary selling price of the retailer.

\item $\overline{q}_{j}\in (0,+\infty )$ is the maximun order size that the
supplier $j$ could satisfy. 
\end{itemize}

A retailer-multi-supplier problem with bounded production is denoted by the tuple $(M,W,C,p,\overline{q})$ where $C=(c_{1},...,c_{m})$, $W=(w_{1},...,w_{m})$, $p$ is the retailer price fucntion and $\overline{q}=\left( \overline{q}_{j}\right)_{j\in M}.$ Given such a problem, let $\mathbb{Q}=\{q\in \mathbb{R}_{+}^{M}\mid q_{M}\leq q^{\ast }$ and $q_{j}\leq \overline{q}_{j}$ for all $j\in M\}$\ be the set of feasible order sizes, that is, those for which the profit unitary selling price is not negative. Notice that $q=\left({q}_{j}\right)_{j\in M},$ with  $q_{j}$ is the retailer order size to supplier $j$ and  $q_{M}:=\sum_{j\in M}q_{j}$. Moreover $\mathbb{Q\neq \emptyset }$ because $q^{\ast }$ is a positive
quantity.

The profit function of the retailer is given by:
\begin{equation*}
\Pi \left( q,W\left( q\right) \right) :=p(q_{M})q_{M}-\sum_{j\in
M}w_{j}(q_{j})q_{j}
\end{equation*}

The order quantity $q$ is determined by the retailer maximizing his profit:
\begin{eqnarray*}
\max &&\Pi \left( q,W\left( q\right) \right) \\
\mbox{s.t.:} &&q\in \mathbb{Q}.
\end{eqnarray*}

Note that $\mathbb{Q}$ is a compact set and $\Pi \left( q,W\left( q\right)
\right) $ is continuous on $\mathbb{Q}$. Hence, there exits always an optimal solution to the above problem. 

 We consider $\mathbb{Q}^{\ast }$ $%
=\{q\in \mathbb{R}_{+}^{M}\mid p(q_{M})\geq \underset{j\in M:q_{j}\neq 0}{%
\max }w_{j}(q_{j})$ and $q_{j}\leq \overline{q}_{j}$ for all $j\in M\}$ to
be the set of feasible order sizes that do not result in losses for the retailer. We know that  
$\mathbb{Q}^{\ast }\mathbb{%
\neq \emptyset }$ because $p(0)>w_{j}(0)$ for all $j\in M$. Moreover, the
reader may notice that any optimal solution $\widehat{q}$ satisfies that $p(%
\widehat{q}_{j})\geq w_{j}(\widehat{q}_{j})$ for all $j\in M$ since
otherwise the retailer incurs in losses$.$ Therefore, the optimal solution
for the above optimization problems is attained in $\mathbb{Q\cap Q}^{\ast }.$
 
We complete the description of a retailer-multi-supplier problem by stating the expression for the profit of supplier $j\in M$. Given the retailer's order size $q\in \mathbb{R}_{+}^{M}$, the profit for
each supplier $j$ is given by:%
\begin{equation*}
\Upsilon _{j}\left( q\right) =\left( w_{j}(q_{j})-c_{j}(q_{j})\right) q_{j}.
\end{equation*}

\section{Cooperation in problems with multiple retailers and suppliers}

We extend now the retailer-multi-suplier problem to the case of multiple retailers  whose demand is replenished by several suppliers. Let  $N\cup M\ $ be the set of agents in the distribution chain, where $N=\{1,...,n\}$ be the set of retailers and $M=\{1,...,m\} $ be the set of suppliers. For the ease of readability, in the following we will use the following
notation.
\smallskip
\begin{itemize}
\item $q_{ij}$ order size of retailer $i\in N$ to supplier $j\in
M$.

\item $q_{Rj}:=\sum_{i\in R}q_{ij},$ total order size by a
coalition $R\subseteq N$ of retailers to supplier $j\in M.$

\item $q_{i}:=\left( q_{i1},...,q_{im}\right) \in \mathbb{R}_{+}^{M}$ and $%
q_{R}:=\left( q_{R1},...,q_{Rm}\right) \in \mathbb{R}_{+}^{M}$ for all $%
R\subseteq N.$

\item $q_{RM}:=\sum_{j\in M}\sum_{i\in R}q_{ij}$.
\end{itemize}

In addition, in this model we assume that each supplier $j$ signs a binding
contract with the retailer $i$ guaranteeing up to a maximum order quantity, $%
\overline{q}_{ij}$, of its production. We collect all of then in $\overline{Q%
}=\left( \overline{q}_{ij}\right) .$ Hence, we can denote as:

\begin{itemize}
\item $\overline{q}_{j}\in (0,+\infty )$ the maximum production capacity for
supplier $j$. That is $\overline{q}_{j}=\sum_{i\in N}\overline{q}_{ij}.$

\item $\overline{q}_{RS}:=\sum_{j\in S}\sum_{i\in R}\overline{q}_{ij}$ the
the maximum order that retailers' coalition $R$ can place to the suppliers'
coalition\ $S.$

\item $\Upsilon _{j}\left( q_{i},q_{R}\right) =\left(
w_{j}(q_{Rj})-c_{j}(q_{Rj})\right) \cdot q_{ij}$ the profit that supplier $%
j\in M$ makes with player $i\in R$ jointly ordering within the coalition $R$%
. Note that $\sum_{i\in R}\Upsilon _{j}\left( q_{i},q_{R}\right) =\Upsilon
_{j}\left( q_{R},q_{R}\right) $ is the profit obtained by supplier $j$ from
the coalition of retailers $R$.
\end{itemize}

A multi-retailer-supplier situation (henceforth MRS-situation) is a tuple $(N,M,W,C,P,\overline{Q})$ where $P=(p_{1},...,p_{n})$, $p_{i}$ is the price function of retailer $i$. By the monotonicity of the price functions, there exists a
value $q_{i}^{\ast }>0$ satisfying that $p_{i}(q_{i}^{\ast })=0,$ for  each $i\in N$.

In such a MRS-situation the cooperation among the agents may
occur in two ways: excluding or including suppliers. First, retailers can place joint orders for the good, getting a discount from the unit price provided by the supplier. This ensures that cooperation between retailers is profitable. Then, the joint profit of a
coalition $R$ equals:%
\begin{eqnarray*}
\max &&\sum_{i\in R}\Pi _{i}\left( q_{i},W\left( q_{R}\right) \right) :=
\left( p_{i}(q_{iM})q_{iM}-\sum_{j\in M}w_{j}(q_{Rj})q_{ij}\right) \\
\mbox{s.t.:} &&q\in \mathbb{Q}^{R},
\end{eqnarray*}%
where 
\begin{equation*}
\mathbb{Q}^{R}:=\left\{ q\in \mathbb{R}_{+}^{\left\vert
R\right\vert \times M}\left\vert 
\begin{array}{c}
q_{iM}\leq q_{i}^{\ast }\text{ for all }i\in R \\ 
q_{ij}\leq \overline{q}_{ij}\text{ for all }j\in M%
\end{array}%
\right. \text{ }\right\} .
\end{equation*}%

$\mathbb{Q}^{R}$ is compact and the objective function $\sum_{i\in R}\Pi
_{i}\left( q_{i},W\left( q_{R}\right) \right) $ is continuous on $\mathbb{Q}%
^{R}.$ We define the feasible set $\mathbb{Q}_{\ast }^{R}\ $ as $\left\{ q\in 
\mathbb{R}_{+}^{\left\vert
R\right\vert \times M}\mid p_{i}\left( q_{iM}\right) \geq \underset{j\in
M:q_{ij}\neq 0}{\max }w_{j}(q_{Rj})\text{ for all }i\in R,\text{ with }\right.$ 
$\left.  q_{ij}\leq \overline{q}_{ij}\text{ for all }j\in M\right\} $. The reader may
notice that all optimal solutions $\widehat{q}\in \mathbb{Q}^{R}\cap \mathbb{Q}%
_{\ast }^{R}.$

A second type of cooperation is among a group of retailers $R$ and suppliers 
$S$. We assume that suppliers share their production technologies, such that
the production function is the minimum of all production functions among
those suppliers belonging to the coalition, i.e, $c_{S}:=\underset{j\in S}{%
\min }\left\{ c_{j}\right\} .$ This can be interpreted as the cooperating suppliers producing their orders jointly by adding their productive capacity. The joint profit of this retailer-supplier
group $\left( R,S\right) $ is:

\begin{gather}  \label{lol}
\sum_{i\in R}\Pi _{i}\left( q_{i},\Psi ^{S}\left( q_{R}\right) \right) 
\notag:=\sum_{i\in R}\left( p_{i}(q_{iM})q_{iM}-\sum_{j\in
M}w_{j}(q_{Rj})q_{ij}+\sum_{j\in S}\Upsilon _{j}^{S}\left(
q_{i},q_{R}\right) \right) \\
=\sum_{i\in R}\left( p_{i}(q_{iM})q_{iM}-\sum_{j\in
S}c_{S}(q_{RS})q_{ij}-\sum_{j\in M\setminus S}w_{j}(q_{Rj})q_{ij}\right) ,
\notag \\
\end{gather}

where%
\begin{eqnarray*}
\Upsilon _{j}^{S}\left( q_{i},q_{R}\right) &:&=\left(
w_{j}(q_{Rj})-c_{S}(q_{RS})\right) q_{ij}; \\
\Psi ^{S}\left( q_{R}\right) &:&=\left[ \left( c_{S}(q_{RS})\right) _{j\in
S},\left( w_{j}\left( q_{Rj}\right) \right) _{j\in M\setminus S}\right] ; \\
\Pi _{i}\left( q_{i},\Psi ^{S}\left( q_{R}\right) \right)
&:&=p_{i}(q_{iM})q_{iM}-\sum_{j\in S}c_{S}(q_{RS})q_{ij}-\sum_{j\in
M\setminus S}w_{j}(q_{Rj})q_{ij}.
\end{eqnarray*}

Thus, the retailers-suppliers group $\left( R,S\right) $ determines their optimal order quantity $q$ as the one that maximizes its profit:

\begin{eqnarray}
\max &&\sum_{i\in R}\Pi _{i}\left( q_{i},\Psi ^{S}\left( q_{R}\right) \right)
\label{P:RS-1} \\
\mbox{s.t.:} && q\in \mathbb{Q}^{R}  \label{P:RS-2}
\end{eqnarray}%

Similar arguments as those used above ensure the existence of an optimal
solution for this optimization problem. Let $q_{i}^{(R,S)}\in \mathbb{R}%
_{+}^{M}$ denote the optimal order size for retailer $i\in R$ when
cooperating with the coalition $S$ of suppliers.

Our next lemma shows the main properties of the retailers' profit function
and the monotonicity property of the order size with respect to the sizes of
the coalitions of suppliers.

\begin{lemma}
\label{Lemma 1} Let $(N,M,W,C,P,\overline{Q})$ be a MRS-situation, and $i\in
R\subseteq N$. Then, for all $S,T\subseteq M$ and such that $T\subseteq S:$

\begin{description}
\item[(P1)] 
\begin{equation*}
\Pi _{i}\left( q_{i}^{(R,S)},\Psi ^{S}\left( q_{R}^{(R,S)}\right) \right)
=\Pi _{i}\left( q_{i}^{(R,S)},\Psi ^{\emptyset }\left( q_{R}^{(R,S)}\right)
\right) +\sum_{j\in S}\Upsilon _{j}^{S}\left(
q_{i}^{(R,S)},q_{R}^{(R,S)}\right) ;
\end{equation*}

\item[(P2)] for all $q\in \mathbb{Q}^{R}:$
\begin{equation*}
\Pi _{i}\left( q_{i},\Psi ^{S}\left( q_{R}\right) \right) \geq \Pi
_{i}\left( q_{i},\Psi ^{T}\left( q_{R}\right) \right);
\end{equation*}

\item[(P3)]
\begin{equation*}
q_{\left(R\setminus \left\{i\right\} \right)M}^{(N,M)}\geq q_{\left(R\setminus \left\{i \right\}\right)M}^{(R,S)};
\end{equation*}

\item[(P4)] $\ $%
\begin{equation*}
\Pi _{i}\left( q_{i}^{(R,S)},\Psi ^{S}\left( q_{R}^{(R,S)}\right) \right)
\geq \Pi _{i}\left( q_{i}^{(R,S)},\Psi ^{\emptyset }\left(
q_{R}^{(R,S)}\right) \right) ;
\end{equation*}

\item[(P5)] 
\begin{equation*}
\sum_{i\in R}\Pi _{i}\left( q_{i}^{(R,S)},\Psi ^{S}\left(
q_{R}^{(R,S)}\right) \right) \geq \sum_{i\in R}\Pi _{i}\left(
q_{i}^{(R,T)},\Psi ^{T}\left( q_{R}^{(R,T)}\right) \right).
\end{equation*}
\end{description}
\end{lemma}

\begin{proof}
(P1) follows inmediatately from (\ref{lol}).

(P2) 
\begin{eqnarray*}
\Pi _{i}\left( q_{i},\Psi ^{S}\left( q_{R}\right) \right)
&=&p_{i}(q_{iM})q_{iM}-\sum_{j\in S}c_{S}(q_{RS})q_{ij}-\sum_{j\in
M\setminus S}w_{j}(q_{Rj})q_{ij} \\
&=&p_{i}(q_{iM})q_{iM}-\sum_{j\in T}c_{S}(q_{RS})q_{ij}-\sum_{j\in
S\backslash T}c_{S}(q_{RS})q_{ij}-\sum_{j\in M\setminus S}w_{j}(q_{Rj})q_{ij}
\\
&\geq &p_{i}(q_{iM})q_{iM}-\sum_{j\in T}c_{S}(q_{RS})q_{ij}-\sum_{j\in
S\backslash T}w_{j}(q_{Rj})q_{ij}-\sum_{j\in M\setminus S}w_{j}(q_{Rj})q_{ij}
\\
&=&p_{i}(q_{iM})q_{iM}-\sum_{j\in T}c_{S}(q_{RS})q_{ij}-\sum_{j\in
M\setminus T}w_{j}(q_{Rj})q_{ij} \\
&\geq &p_{i}(q_{iM})q_{iM}-c_{T}(q_{RT})q_{iT}-\sum_{j\in M\setminus
T}w_{j}(q_{Rj})q_{ij} \\
&=&\Pi _{i}\left( q_{i},\Psi ^{T}\left( q_{R}\right) \right)
\end{eqnarray*}

since $c_{S}(q_{RS})\leq c_{S}(q_{RT})\leq c_{T}(q_{RT})\ $and $q_{RS}\geq
q_{RT}.$

(P3) $q_{\left(R\setminus \left\{i\right\} \right)M}^{(N,M)}\geq q_{\left(R\setminus \left\{i \right\}\right)M}^{(R,S)}$ since $R\subseteq N$ and $S\subseteq M$, costs are
cheaper in the problem involving $\left(N,M\right)$ than in the one involving $\left(R,S\right)$.
Therefore, the total order of coalition $R\setminus \left\{i\right\}$ cannot be smaller.

(P4) follows from (P1) and positivity of functions $\Pi _{i}$ and $\Upsilon
_{j}^{S}$ for all $i\in R\subseteq N$ and $j\in S\subseteq M.$

Finally, (P5):%
\begin{equation*}
\sum_{i\in R}\Pi _{i}\left( q_{i}^{(R,T)},\Psi ^{T}\left(
q_{R}^{(R,T)}\right) \right) \leq \sum_{i\in R}\Pi _{i}\left(
q_{i}^{(R,T)},\Psi ^{S}\left( q_{R}^{(R,T)}\right) \right) \leq \sum_{i\in
R}\Pi _{i}\left( q_{i}^{(R,S)},\Psi ^{S}\left( q_{R}^{(R,S)}\right) \right)
\end{equation*}

Note that the first inequality follows from (P2) and the last one is true
since in general $q^{(R,T)}\in \mathbb{Q}^R$ and $q^{(R,S)}\in \mathbb{Q}^R$ is the
optimal solution.   \hfill
\end{proof}
\medskip
Next we define the cooperative game corresponding to a MRS-situation.

\begin{definition}
Let $(N,M,W,C,P,\overline{Q})$ be a MRS-situation. The corresponding
MRS-game $(N,M,v)$ with characteristic function $v:N\times M\rightarrow 
\mathbb{R}$ is given by 
\begin{equation*}
v(R,S)=\sum_{i\in R}\Pi _{i}\left( q_{i}^{(R,S)},\Psi ^{S}\left(
q_{R}^{(R,S)}\right) \right)
\end{equation*}%
for all coalitions $R\subseteq N$, $S\subseteq M$ and $v(\emptyset ,S)=0$
for all $S\subseteq M.$
\end{definition}

The definition of MRS-situations and their corresponding games focuses on
retailer revenues arising from selling goods to consumers. The following
result shows some characteristics of MRS-games.

\begin{proposition}
\label{lemma}Let $(N,M,W,C,P,\overline{Q})$ be a MRS-situation and $(N,M,v)$
its corresponding MRS-game. Then

\begin{description}
\item[(i)] $v(R,S)>0$ for all coalitions $\emptyset \neq R\subseteq N$ and 
$S\subseteq N$.
\item[(ii)] $v(R,S)+v(F,T) \le v(R\cup F, S\cup T)$ for any $R,F\subseteq N$ and $S,T\subseteq M$ such that $R\cap F = \emptyset$ and  $S\cap T = \emptyset$; i.e.,  $v$ is superadditive.
\item[(iii)] $v(R,S)\le v(F,T)$ for any $R\subseteq F$ and $S\subseteq   T$.
\end{description}
\end{proposition}
\begin{proof}
(i) Follows inmediatately from the definition of the game and the positive
profit margins for the retailers. (ii) Let $R,F\subseteq N$ and $%
S,T\subseteq M$ be disjoint coalitions of retailer and suppliers
respectively. By definition of the MRS-game%
\begin{eqnarray*}
v(R,S)+v(F,T) &=&\sum_{i\in R}\Pi _{i}\left( q_{i}^{(R,S)},\Psi ^{S}\left(
q_{R}^{(R,S)}\right) \right) \\
&&+\sum_{i\in F}\Pi _{i}\left( q_{i}^{(F,T)},\Psi ^{T}\left(
q_{F}^{(F,T)}\right) \right).
\end{eqnarray*}%
Define $\widehat{q}_{ij}=q_{ij}^{(R,S)}$ if $i\in R$, $j\in S$ and $\widehat{%
q}_{ij}=q_{ij}^{(F,T)}$ if $i\in F$, $j\in T.$ Consider $H=R\cup F.$ Now
applying (P2) of Lemma \ref{Lemma 1} 
\begin{eqnarray*}
v(R,S)+v(F,T) &\leq &\sum_{i\in R}\Pi _{i}\left( \widehat{q}_{i},\Psi
^{S}\left( \widehat{q}_{H}\right) \right) +\sum_{i\in F}\Pi _{i}\left( 
\widehat{q}_{i},\Psi ^{T}\left( \widehat{q}_{H}\right) \right) \\
&\leq &\sum_{i\in R}\Pi _{i}\left( \widehat{q}_{i},\Psi ^{S\cup T}\left( 
\widehat{q}_{H}\right) \right) +\sum_{i\in F}\Pi _{i}\left( \widehat{q}%
_{i},\Psi ^{S\cup T}\left( \widehat{q}_{H}\right) \right) \\
&\leq &\sum_{i\in H}\Pi _{i}\left( q_{i}^{(H,S\cup T)},\Psi ^{S\cup T}\left(
q_{H}^{(H,S\cup T)}\right) \right) =v(R\cup F,S\cup T)
\end{eqnarray*}%
since $\widehat{q}_{i}$ is not necessary optimal for $(H,S\cup T).$ Finally,
(iii) follows from (i) and (ii).  \hfill
\end{proof}

\bigskip Property (i) shows that there is always a benefit from cooperation between retailers and suppliers. Property (ii) means that MRS-games are superadditive. Consequently, it makes sense for the grand coalition to form. We then consider how to allocate the profit of this all included coalition among suppliers and retailers. 

In the following we show that the class of MRS-games is balanced. Let $(N,M,W,C,P,\overline{Q})$ be a MRS-situation, the core of the corresponding MRS-game $(N,M,v)$ is defined as follows:

\begin{equation*}
Core(N,M,v)=\left\{ x\in \mathbb{R}^{\left\vert N\cup M\right\vert
}\left\vert 
\begin{array}{c}
\sum_{i\in N\cup M}x_{i}=v(N,M); \\ 
\sum_{i\in R\cup S}x_{i}\geq v(R,S),\forall S\subseteq N,\forall R\subseteq M%
\end{array}%
\right. \right\} .
\end{equation*}

Before showing that $Core(N,M,v) \neq \emptyset,$ we prove the following technical lemma.

\begin{lemma}
\label{lemma2} \label{Lemma2}Let $(N,M,W,C,P,\overline{Q})$ be a MRS-situation,
and $i\in R\subseteq N$. Then, for all $S\subseteq M:$%
\begin{equation*}
\Pi _{i}\left( q_{i}^{(N,M)},\Psi ^{M}\left( q_{N}^{(N,M)}\right) \right)
\geq \Pi _{i}\left( q_{i}^{(R,S)},\Psi ^{S}\left( q_{R}^{(R,S)}\right)
\right) ;
\end{equation*}
\end{lemma}

\begin{proof}
Using the notation stated above we have that: 
\begin{align*}
\Pi _{i}\left( q_{i}^{(N,M)},\Psi ^{M}\left( q_{N}^{(N,M)}\right) \right)
=& \max_{q_{i}}\Pi _{i}\left( q_{i},\Psi ^{M}\left( q_{N\backslash
R}^{(N,M)}+q_{i}+q_{R\backslash \{i\}}^{(N,M)}\right) \right)  \\
(\mbox{since }  q_{N\setminus R}^{(N,M)}\ge 0)  \geq &\max_{q_{i}}\Pi _{i}\left( q_{i},\Psi ^{M}\left( q_{i}+q_{R\backslash
\{i\}}^{(N,M)}\right) \right)  \\
(\mbox{since } q_i^{(R,S)}  \mbox{not necessarily optimal})  \geq &\Pi _{i}\left( q_{i}^{(R,S)},\Psi ^{M}\left(
q_{i}^{(R,S)}+q_{R\backslash \{i\}}^{(N,M)}\right) \right)  \\
(\mbox{by \textbf{P3} in Lemma \ref{Lemma 1}})    \geq &\Pi _{i}\left( q_{i}^{(R,S)},\Psi ^{M}\left( q_{R}^{(R,S)}\right)
\right)  \\
(\mbox{by \textbf{P2} in Lemma \ref{Lemma 1}}) \geq &\Pi _{i}\left( q_{i}^{(R,S)},\Psi ^{S}\left( q_{R}^{(R,S)}\right)
\right) .
\end{align*}
\hfill \end{proof}

\begin{theorem}
\label{th:5.5} Let $(N,M,W,C,P,\overline{Q})$ be a MRS-situation and $(N,M,v)$
be the corresponding MRS-game. Then, $(N,M,v)$ is balanced.
\end{theorem}
\begin{proof}
The allocation $x^{a}(v)$ defined by 
$x_{i}^{a}(v):=\Pi _{i}\left( q_{i}^{(N,M)},\Psi ^{M}\left(
q_{N}^{(N,M)}\right) \right) ,$ for all $i\in N $ and $x_{j}^{a}(v):=0,$ for all $j\in M$. In this allocation the retailer receives all the profit from cooperation with the suppliers,
whereas each supplier receives nothing. Note that

\begin{equation*}
\sum_{i\in N\cup M}x_{i}^{a}(v)=\sum_{i\in N}\Pi _{i}\left(
q_{i}^{(N,M)},\Psi ^{M}\left( q_{N}^{(N,M)}\right) \right) =v(N,M).
\end{equation*}

Hence, $x^{a}(v)$ is an efficient allocation. Consider coalitions $%
\emptyset \neq R\subseteq N$ and $S\subseteq M.$ Then

\begin{eqnarray*}
\sum_{i\in R\cup S}x_{i}^{a}(v) &=&\sum_{i\in R}\Pi _{i}\left(
q_{i}^{(N,M)},\Psi ^{M}\left( q_{N}^{(N,M)}\right) \right) \\
&\geq &\sum_{i\in R}\Pi _{i}\left( q_{i}^{(R,S)},\Psi ^{S}\left(
q_{R}^{(R,S)}\right) \right) =v(R,S),
\end{eqnarray*}%
where the last inequality follows from Lemma \ref{Lemma2}. Finally, $%
v(\emptyset ,S)=0=\sum_{i\in S}x_{i}^{a}(v)$ for all $S\subseteq M.$
Therefore, we conclude that the game is balanced. \hfill 
\end{proof}

\bigskip The allocation $x^{a}(v)$ proposed in the proof ot the previous theorem, that will call from now on the altruistic allocation, is inspired by the one used in Guardiola et al. (2007). The following example illustrates a MRS-situation with two retailers and two suppliers.

\begin{example}
\label{Examp2}Consider the next example with 2 retailers and 2 suppliers
where: $p_{1}(q)=5-\frac{q}{10}$ $,p_{2}(q)=8-\frac{q}{5}$ and $\overline{q}%
_{11}=\overline{q}_{12}=\overline{q}_{21}=\overline{q}_{22}=10.$%
\begin{equation*}
w_{1}(q)=\left\{ 
\begin{array}{cc}
4-\frac{q}{20}, & 0\leq q\leq 20, \\ 
3, & q>20%
\end{array}%
\right. \text{ \ \ \ \ }c_{1}(q)=\left\{ 
\begin{array}{cc}
3-\frac{q}{20}, & 0\leq q\leq 30, \\ 
1.5, & q>30%
\end{array}%
\right. 
\end{equation*}%
\begin{equation*}
w_{2}(q)=\left\{ 
\begin{array}{cc}
3-\frac{q}{30}, & 0\leq q\leq 30, \\ 
2, & q>30.%
\end{array}%
\right. \text{ \ \ \ \ }c_{2}(q)=\left\{ 
\begin{array}{cc}
2-\frac{q}{30}, & 0\leq q\leq 40, \\ 
\frac{2}{3}, & q>40%
\end{array}%
\right. 
\end{equation*}%
\begin{equation*}
c_{\{1,2\}}(q)=\left\{ 
\begin{array}{cc}
2-\frac{q}{30}, & 0\leq q\leq 40, \\ 
\frac{2}{3}, & q>40%
\end{array}%
\right. 
\end{equation*}

It follows that $q_{1}^{\ast }=50,q_{2}^{\ast }=40.$ Solving the
corresponding optimization problems for the different coalitions $(R,S)$ we
obtain the following table:%
\begin{table}[!h]
\begin{center}
\begin{tabular}{|c|c|c|c|c|} 
 \hline
 $R$ & $S$ & $q_{1}^{(R,S)}$ & $q_{2}^{(R,S)}$ & $v(R,S)$ \\ [0.5ex] 
 \hline\hline
$\{1\}$ & $\emptyset$  & $\left( 0,10\right)$  & - & $13\frac{1}{3}$ \\ \hline
$\{1\}$ & $\{1\}$ & $\left( 10,0\right)$  & - & $15$ \\ \hline
$\{1\}$ & $\{2\}$ & $\left( 0,10\right)$  & - & $23\frac{1}{3}$ \\ \hline
$\{1\}$ & $\{1,2\}$ & $\left( 10,10\right)$  & - & $33\frac{1}{3}$ \\ \hline
$\{2\}$ & $\emptyset$  & - & $\left( 0,10\right)$  & $33\frac{1}{3}$ \\ \hline
$\{2\}$ & $\{1\}$ & - & $\left( 10,3\right)$  & $36\frac{1}{2}$ \\ \hline
$\{2\}$ & $\{2\}$ & - & $\left( 0,10\right)$  & $43\frac{1}{3}$ \\ \hline
$\{2\}$ & $\{1,2\}$ & - & $\left( 9,9\right)$  & $54$ \\ \hline
$\{1,2\}$ & $\emptyset$  & $\left( 0,10\right)$  & $\left( 0,10\right)$  & $53\frac{1}{3}$ \\ \hline
$\{1,2\}$ & $\{1\}$ & $\left( 10,1\frac{2}{3}\right)$  & $\left( 10,3\frac{1}{3}\right)$  & $61\frac{2}{3}$ \\ \hline
$\{1,2\}$ & $\{2\}$ & $\left( 0,10\right)$  & $\left( 0,10\right)$  & $73\frac{1}{3}$\\ \hline
$\{1,2\}$ & $\{1,2\}$ & $\left( 10,10\right)$  & $\left( 10,10\right)$  & $113\frac{1}{3}$ \\ 
\hline
\end{tabular}
\caption{\label{table 1}}
\end{center}
\end{table}

The reader can easily check that $x^{a}(v)=\left( 46\frac{2}{3},66\frac{2}{3},0,0\right) $ belongs to the core.
\end{example}

In the next section, we propose and characterize a new allocation for these MRS-games which compensates suppliers with some benefit for their cooperation.

\section{An allocation rule that compensates suppliers}

The previous section refers to the altruistic allocation as a single-valued aloocation for MRS-games. This allocation is always in the core, but it may be questioned because it gives some suppliers a zero payoff, even though they can be  necessary for the highest total profits. Our goal is to identify solutions that recognize the value of all the suppliers who contribute to get the greatest joint production. 

This section is devoted to provide a more desirable allocation for all the
suppliers since with the altruistic allocation some of them receive a null payment even though they are necessary for the grand coalition to obtain the final benefit.
To do so, let us take the altruistic allocation as a starting point and from
there, we will build an alternative single-valued solution $\varphi $ taking 
into account the following five desirable properties on the class of MRS-games $(N,M,v)$:

\begin{description}
\item[(EF)] Efficiency. $\sum_{i\in NUM}\varphi _{i}(v)=v(N,M).$

\item[(SR)] Stability for retailers. $\sum_{i\in R}\varphi _{i}(v)\geq
v(R,M\backslash \{j\})$ for all coalitions $R\subseteq N$ and all $j\in M.$

\item[(RR)] Retailer reduction. $\varphi _{i}(v)=x_{i}^{a}(v)-\frac{%
\sum_{k\in R}x_{k}^{a}(v)-v(R,M\backslash \{j\})}{\left\vert R\right\vert }$
for some coalition $\emptyset \neq R\subseteq N$ with $i\in R$ and some $%
j\in M.$

\item[(PD)] Preservation of differences for retailers. $\varphi
_{i}(v)-\varphi _{i^{\prime }}(v)=x_{i}^{a}(v)-x_{i^{\prime }}^{a}(v)$ for
all $i,i^{\prime }\in N$.

\item[(PP)] Proportionality to the production. $\varphi _{j}(v)\cdot
q_{Nj^{\prime }}^{(N,M)}=\varphi _{j^{\prime }}(v)\cdot q_{Nj}^{(N,M)}$ for
all $j,j^{\prime }\in M$ and for some optimal solution $q^{(N,M)}.$
\end{description}

\noindent \textit{Efficiency} implies that the overall surplus induced by
cooperation is allocated among the players. While 
\textit{stability for retailers} ensures that any coalition of retailers is
coalitionally stable with respect to the coalitions they form in which a supplier is missing. The \textit{retailer reduction } property states that
the amount received by each retailer is less than the amount allocated to it
by the altruistic solution. This reduction equals $\left( \sum_{i\in
R}x_{i}^{a}(v)-v(R,M\backslash \{j\})\right) /\left\vert R\right\vert ,$
which is the per capita gain for the retailers coalition $R$ from cooperation
with any coalition in which a supplier is missing. The\textit{\ preservation of differences for
retailers} property is a modification of the \textit{preservation of
differences} property by Hart and Mas-Colell (1989). This property ensures that the
difference in the allocation that assigned the altruistic solution to two
retailers must be maintained. Finally, \textit{%
proportionality to the production } property requires that the allocation of any
two suppliers must be in the same proportion as the quantities they produce.

The main result in this section states that there exists a unique solution
for MRS-games satisfying the properties (EF), (SR), (RR), (PD) and (PP).

\begin{theorem}
Let $(N,M,W,C,P,\overline{Q})$ be a MRS-situation, $(N,M,v)$ the
corresponding MRS-game and $q^{(N,M)}$ an optimal solution. The unique
solution $\xi $ on the class of MRS-games, satisfying (EF), (SR), (RR), (PD)
and (PP) is $\xi (v)=\left( \xi _{i}(v)\right) _{i\in N\cup M}$ defined by 
\begin{equation*}
\xi _{i}(v):=\left\{ 
\begin{array}{ll}
\Pi _{i}\left( q_{i}^{(N,M)},\Psi ^{M}\left( q_{N}^{(N,M)}\right) \right)
-\beta , & i\in N, \\ 
&  \\ 
\frac{q_{Ni}^{(N,M)}\cdot \left\vert N\right\vert \cdot \beta }{%
q_{NM}^{(N,M)}}, & i\in M,%
\end{array}%
\right. 
\end{equation*}%
where $\beta :=\underset{\emptyset \neq R\subseteq N,j\in M}{\min }%
\left\{ \frac{\sum_{i\in R}x_{i}^{a}(v)-v(R,M\backslash \{j\})}{\left\vert
R\right\vert }\right\} .$
\end{theorem}

\begin{proof}
It is clear that $\xi (v)$ satisfies (EF), (SR), (RR), (PD) and (PP).

To show the converse, take an optimal solution $q^{(N,M)}$and a solution $%
\varphi $ on the class of MRS-games that satisfies (EF), (SR), (RR), (PD)
and (PP). By (RR), $\varphi _{i}(v)=x_{i}^{a}(v)-\alpha _{i}$ with $\alpha
_{i}=\frac{1}{\left\vert R\right\vert }\left( \sum_{k\in
R}x_{k}^{a}(v)-v(R,M\backslash \{j\})\right) $ for some coalition $\emptyset
\neq R\subseteq N$, for each retailer $i$ and some $j\in M$. By (PD), $%
\alpha _{i}=\alpha _{i^{\prime }}$ for all $i,i^{\prime }\in N$ with $i\neq
i^{\prime }$. This implies $\alpha _{i}=\alpha _{\ast }$ for all $i\in N$.
According to (SR) $\sum_{i\in R}\varphi _{i}(v)=\sum_{i\in
R}x_{i}^{a}(v)-\left\vert R\right\vert \cdot \alpha _{\ast }\geq
v(R,M\backslash \{j\})$ for all coalitions $R\subseteq N$ and all $j\in M$,
or equivalently $\alpha _{\ast }\leq \frac{1}{\left\vert R\right\vert }%
\left( \sum_{i\in R}x_{i}^{a}(v)-v(R,M\backslash \{j\})\right) $ for all
coalition $R\subseteq N$ and all $j\in M$. But then, $\alpha _{\ast }=\beta
=\min_{\emptyset \neq R\subseteq N,j\in M}\{\left( \sum_{i\in
R}x_{i}^{a}(v)-v(R,M\backslash \{j\})\right) /\left\vert R\right\vert \}$.
Moreover, there is at least one supplier from whom an order is placed ($%
q_{NM}^{(N,M)}$ $\neq 0$) . Therefore, we consider $j\in M$ such that $%
q_{Nj}^{(N,M)}$ $\neq 0$, then by (PP) $\varphi _{k}(v)=\frac{\varphi
_{j}(v)\cdot q_{Nk}^{(N,M)}}{q_{Nj}^{(N,M)}}$ for all $k\in M.$ Finally, by
(EF) $\sum_{i\in N}\varphi _{i}(v)+\sum_{k\in M}\varphi _{k}(v)=\sum_{i\in
N}x_{i}^{a}(v)-\left\vert N\right\vert \cdot \beta +\sum_{k\in M}\frac{%
\varphi _{j}(v)\cdot q_{Nk}^{(N,M)}}{q_{Nj}^{(N,M)}}=v(N,M)$ and we obtain
that $\varphi _{j}(v)=\frac{q_{Nj}^{(N,M)}\cdot \left\vert N\right\vert
\cdot \beta }{q_{NM}^{(N,M)}}.$ If $q_{Nj}^{(N,M)}$ $=0$ then by (PP) $%
\varphi _{j}(v)=0.$ Thus, we conclude\ $\varphi =\xi $.\hfill 
\end{proof}

\bigskip The reader may notice that for a MRS-situation $(N,M,W,C,P,%
\overline{Q})$ there exists a unique MRS-game $(N,M,v)$ although we take
several different optimal solutions. However, $\xi (v)$ could change
depending on the choice of those optimal solutions. In the allocation, retailers compensate to each suppliers with the same amount $\left\vert N\right\vert \cdot \beta ,$
weighted by their contribution to total production. This allocation $\xi (v)$ is called Supplier Compensation allocation (henceforth SC-allocation). The following proposition shows that the SC-allocation is a core allocation regardless of the optimal solution.

\begin{proposition}
Let $(N,M,W,C,P,\overline{Q})$ be a MRS-situation, $(N,M,v)$ the
corresponding MRS-game and $q^{(N,M)}$ an optimal solution. Then, $\xi (v)\in Core(N,M,v)$.
\end{proposition}

\begin{proof}
It is clear that $\xi (v)$ satisfies (EF) and (SR). Therefore, we have that $%
\sum_{i\in R\cup S}\xi _{i}(v)\geq \sum_{i\in R}\xi _{i}(v)\geq
v(R,M\backslash \{j\}),$ for all supplier $j$. Hence, $\sum_{i\in R\cup S}\xi _{i}(v)\geq v(R,S)$
for all $R\subseteq N$ and for all $S\varsubsetneq M.$ Finally, 

\begin{eqnarray*}
\sum_{i\in R\cup M}\xi _{i}(v) &=&\sum_{i\in R}x_{i}^{a}(v)-\left\vert
R\right\vert \cdot \beta +\sum_{j\in M}\frac{q_{Nj}^{(N,M)}\cdot \left\vert
N\right\vert \cdot \beta }{q_{NM}^{(N,M)}} \\
&=&\sum_{i\in R}x_{i}^{a}(v)-\left\vert R\right\vert \ \cdot \beta
+\left\vert N\right\vert \cdot \beta  \\
&\geq&\sum_{i\in R}x_{i}^{a}(v)\geq v(R,M).
\end{eqnarray*}

We conclude that $\xi (v)\in Core(N,M,v).$\hfill 
\end{proof}

\bigskip In the following example we compute the MRS-allocation for the data in Example \ref{Examp2} and compare it with the altruistic allocation.

\begin{example}
Consider again the data in Example \ref{Examp2}. Recall that the benefit of the grand coalition is $V(N,M)=113\frac{1}{3}$. We compare the SC-allocation with the altruistic allocation for this example.

\begin{center}
\begin{tabular}{c|c}
$x^{a}(v)$ & $\xi (v)$ \\ \hline
&  \\ 
$\left( 
\begin{array}{c}
46\frac{2}{3} \\ 
66\frac{2}{3} \\ 
0 \\ 
0%
\end{array}%
\right) $ & $\left( 
\begin{array}{c}
26\frac{2}{3} \\ 
46\frac{2}{3} \\ 
20 \\ 
20%
\end{array}%
\right) $%
\end{tabular}
\end{center}

\noindent with $\beta :=\underset{\emptyset \neq R\subseteq N,j\in M}{\min }\left\{ 
\frac{\sum_{i\in R}x_{i}^{a}(v)-v(R,M\backslash \{j\})}{\left\vert
R\right\vert \ }\right\} =\min \{46\frac{2}{3}-23\frac{1}{3},46\frac{2}{3}%
-15,66\frac{2}{3}-43\frac{1}{3},66\frac{2}{3}-36\frac{1}{2},\frac{113\frac{1%
}{3}-73\frac{1}{3}}{2},\frac{113\frac{1}{3}-61\frac{2}{3}}{2}\}=20.$

The reader may notice that, as we have already pointed out, the
MRS-allocation reduces the retailers' profit $(20$,$20)$ to give this profit
to the suppliers
\end{example}

To conclude this section, the five examples below illustrate that the properties (EF), (SR), (RR), (PD) and (PP) are logically independent.

\begin{example} (EF fails)
Consider the single-solution $\varphi $ on the class of RMS-games defined by 
\begin{equation*}
\varphi _{i}(v):=\left\{ 
\begin{array}{ll}
\Pi _{i}\left( q_{i}^{(N,M)},\Psi ^{M}\left( q_{N}^{(N,M)}\right) \right)
-\beta , & i\in N, \\ 
0, & i\in M,%
\end{array}%
\right. 
\end{equation*}%
where $\beta :=\underset{\emptyset \neq R\subseteq N,j\in M}{\min }\left\{ 
\frac{\sum_{i\in R}x_{i}^{a}(v)-v(R,M\backslash \{j\})}{\left\vert
R\right\vert }\right\} .$ $\varphi $ satisfies (SR), (RR), (PD) and (PP) but
not (EF).
\end{example}

\begin{example} (SR fails)
Let a single-solution $\varphi $ on the class of RMS-games be defined by 
\begin{equation*}
\varphi _{i}(v):=\left\{ 
\begin{array}{ll}
\Pi _{i}\left( q_{i}^{(N,M)},\Psi ^{M}\left( q_{N}^{(N,M)}\right) \right)
-\beta ^{\ast }, & i\in N, \\ 
\frac{q_{Ni}^{(N,M)}\cdot \left\vert N\right\vert \cdot \beta ^{\ast }}{%
q_{NM}^{(N,M)}}, & i\in M,%
\end{array}%
\right. 
\end{equation*}%
where $\beta ^{\ast }:=\underset{\emptyset \neq R\subseteq N,j\in M}{\max }%
\left\{ \frac{\sum_{i\in R}x_{i}^{a}(v)-v(R,M\backslash \{j\})}{\left\vert
R\right\vert }\right\} .$ $\varphi $ satisfies (EF), (RR), (PD) and (PP) but
not (SR).
\end{example}

\begin{example} (RR fails)
We take a single-solution $\varphi $ on the class of RMS-games defined by 
\begin{equation*}
\varphi _{i}(v):=\left\{ 
\begin{array}{ll}
\Pi _{i}\left( q_{i}^{(N,M)},\Psi ^{M}\left( q_{N}^{(N,M)}\right) \right)
-\beta ^{\ast }, & i\in N, \\ 
\frac{q_{Ni}^{(N,M)}\cdot \left\vert N\right\vert \cdot \beta ^{\ast }}{%
q_{NM}^{(N,M)}}, & i\in M,%
\end{array}%
\right. 
\end{equation*}%
where $\beta ^{\ast }:=\underset{\emptyset \neq R\subseteq N,j\in M}{\min }%
\left\{ \frac{\sum_{i\in R}x_{i}^{a}(v)-v(R,M\backslash \{j\})}{2\cdot
\left\vert R\right\vert }\right\} .$ $\varphi $ satisfies (EF), (SR), (PD)
and (PP) but not (RR).
\end{example}

\begin{example} (PD fails)
Take a single-solution $\varphi $ on the class of RMS-games defined by 
\begin{equation*}
\varphi _{i}(v):=\left\{ 
\begin{array}{ll}
\Pi _{i}\left( q_{i}^{(N,M)},\Psi ^{M}\left( q_{N}^{(N,M)}\right) \right)
-\beta _{i}, & i\in N, \\ 
\frac{q_{Ni}^{(N,M)}\cdot \sum_{k\in N}\beta _{k}}{q_{NM}^{(N,M)}}, & i\in M,%
\end{array}%
\right. 
\end{equation*}%
with $\beta _{i}:=\underset{j\in M}{\min }\left\{ \frac{\sum_{i\in R^{\ast
}}x_{i}^{a}(v)-v(R^{\ast },M\backslash \{j\})}{\left\vert R^{\ast
}\right\vert }\right\} $ if $i\in N\backslash \{k\}$ where $R^{\ast }$ is
the coalition that minimizes $\frac{\sum_{i\in
R}x_{i}^{a}(v)-v(R,M\backslash \{j\})}{\left\vert R\right\vert }$ for all $%
R\subseteq N$ and for all $j\in M$. Moreover, retailer $k$ satisfies that $%
k\notin R^{\ast }$ and $\beta _{k}=\frac{\sum_{i\in
F}x_{i}^{a}(v)-v(F,M\backslash \{j\})}{\left\vert F\right\vert }$ with $%
F\neq R^{\ast }$ but satisfying that $\sum_{i\in R\cup \{k\}}\varphi
_{i}(v)\geq v(R\cup \{k\},M\backslash \{j\})$ for all coalitions $R\subseteq
N\backslash \{k\}$ and all $j\in M.$ $\varphi $ satisfies (EF), (SR), (RR)
and (PP) but not (PD).
\end{example}

\begin{example} (PP fails)
Define $\varphi $ on the class of RMS-games by 
\begin{equation*}
\varphi _{i}(v):=\left\{ 
\begin{array}{ll}
\Pi _{i}\left( q_{i}^{(N,M)},\Psi ^{M}\left( q_{N}^{(N,M)}\right) \right)
-\beta , & i\in N, \\ 
\frac{\left\vert N\right\vert \cdot \beta }{\left\vert M\right\vert }, & 
i\in M,%
\end{array}%
\right. 
\end{equation*}%
where $\beta :=\underset{\emptyset \neq R\subseteq N,j\in M}{\min }\left\{ 
\frac{\sum_{i\in R}x_{i}^{a}(v)-v(R,M\backslash \{j\})}{\left\vert
R\right\vert }\right\} .$ $\varphi $ satisfies (EF), (SR), (RR) and (PD) but
not (PP).
\end{example}

\section{MRS-games with unbounded production}

It is important for the reader to note that retailers' ordering decisions are influenced by suppliers' production limits. Therefore, an important, realistic case of MRS-situations that  should  be analyzed is that  in which the production limits do not affect the choice of optimal solutions in respective optimization problems, i.e., $\overline{q}%
_{ij}=K$, for all $i\in N$ and $j\in M,$ with $K\in \mathbb{R}$ large
enough. We will denote this unbounded situation as $(N,M,W,C,P,\infty).$ The
following example illustrates an unbounded situation with two retailers and
two suppliers.

\begin{example}
\label{Exep1}Consider a MRS-situation with $2$ retailers and 2 suppliers: $%
p_{1}(q)=30-\frac{q}{10},p_{2}(q)=27-\frac{q}{40}$, $\overline{Q}$ values large enough,%
\begin{equation*}
w_{1}(q)=\left\{ 
\begin{array}{cc}
20-\frac{q}{160}, & 0\leq q\leq 1600, \\ 
10, & q>1600%
\end{array}%
\right. \text{ \ \ \ \ }c_{1}(q)=\left\{ 
\begin{array}{cc}
15-\frac{q}{160}, & 0\leq q\leq 1200, \\ 
7.5, & q>1200.%
\end{array}%
\right.
\end{equation*}%
\begin{equation*}
w_{2}(q)=\left\{ 
\begin{array}{cc}
25-\frac{q}{110}, & 0\leq q\leq 1375, \\ 
12.5, & q>1375.%
\end{array}%
\right. \text{ \ \ \ \ }c_{2}(q)=\left\{ 
\begin{array}{cc}
20-\frac{q}{110}, & 0\leq q\leq 1100, \\ 
10, & q>1100.%
\end{array}%
\right.
\end{equation*}

Then, we have that $q_{1}^{\ast }=300,q_{2}^{\ast }=1080.$ Solving the
corresponding optimization problems for the different $(R,S)$ pairs of
retailers-suppliers coalitions, we obtain the following table:%

\begin{table}[!h]
\begin{center}
\begin{tabular}{|c|c|c|c|c|} 
 \hline
 $R$ & $S$ & $q_{1}^{(R,S)}$ & $q_{2}^{(R,S)}$ & $v(R,S)$ \\ [0.5ex] 
 \hline\hline
$\{1\}$ & $\emptyset$ & $\left( 53\frac{1}{3},0\right)$ & - & $266\frac{2}{3}$ \\ [0.5ex] \hline
$\{1\}$ & $\{1\}$ & $\left( 80,0\right)$ & - & $600$ \\ [0.5ex] \hline
$\{1\}$ & $\{2\}$ & $\left( 0,55\right)$ & - & $275$ \\ [0.5ex] \hline
$\{1\}$ & $\{1,2\}$ & $\left( 80,0\right)$ & - & $600$ \\ [0.5ex] \hline
$\{2\}$ & $\emptyset$ & - & $\left( 186\frac{2}{3},0\right)$ & $653\frac{1}{3}$ \\ [0.5ex] \hline
$\{2\}$ & $\{1\}$ & - & $\left( 320,0\right)$ & $1920$ \\ [0.5ex] \hline 
$\{2\}$ & $\{2\}$ & - & $\left( 0,220\right)$ & $770$ \\ [0.5ex] \hline 
$\{2\}$ & $\{1,2\}$ & - & $\left( 320,0\right)$ & $1920$ \\ [0.5ex] \hline 
$\{1,2\}$ & $\emptyset$ & $\left( 67\frac{3}{11},0\right)$ & $\left( 209\frac{1}{11},0\right)$ & $1036\frac{2}{3}$ \\ [0.5ex] \hline
$\{1,2\}$ & $\{1\}$ & $\left( 103\frac{7}{11},0\right)$ & $\left( 354\frac{6}{11},0\right)$ & $2904\frac{1}{2}$ \\ [0.5ex] \hline
$\{1,2\}$ & $\{2\}$ & $\left( 0,81\frac{2}{3}\right)$ & $\left( 0,266\frac{2}{3}\right)$ & $1342\frac{1}{3}$ \\ [0.5ex] \hline
$\{1,2\}$ & $\{1,2\}$ & $\left( 103\frac{7}{11},0\right)$ & $\left( 354\frac{6}{11},0\right)$ & $2904\frac{1}{2}$ \\ [0.5ex] \hline
\end{tabular}
\caption{\label{table 3}}
\end{center}
\end{table}

Note that $v(N,M)=v(\{1,2\},\{1\})=2904\frac{1}{2}$ (i.e., $2904.5$). This
means that adding the supplier $2$ does not increase the profit of the group. 
We could say that supplier $1$ is an optimal supplier.
\end{example}

Formally, an optimal supplier is
defined as $j\in M$ such that $v(N,M)=v(N,\{j\})$. We denote by $M^{o}$, the
set of all optimal suppliers.

The following proposition shows that there is always an optimal supplier in
a MRS-game with unbounded production capacity.

\begin{lemma}
\label{lml}
Let $(N,M,W,C,P,\infty )$ be a MRS-situation. The corresponding MRS-game $%
(N,M,v)$ satisfies that $\left\vert M^{o}\right\vert \geq 1$.
\end{lemma}
\begin{proof}
Take an optimal solution $q^{(N,M)}$ and select the supplier $k\in M$ such
that $c_{k}(q_{NM}^{(N,M)})\leq c_{j}(q_{NM}^{(N,M)})$ for all $j\in
M.$ Then, 
\begin{eqnarray*}
v(N,M) &=&\sum_{i\in N}\Pi _{i}\left( q_{i}^{(N,M)},\Psi ^{M}\left(
q_{N}^{(N,M)}\right) \right) =\sum_{i\in N}\Pi _{i}\left( q_{i}^{(N,M)},\Psi
^{\{k\}}\left( q_{N}^{(N,M)}\right) \right) \\
&=&\sum_{i\in N}\Pi _{i}\left( q_{i}^{(N,\{k\})},\Psi ^{\{k\}}\left(
q_{N}^{(N,\{k\})}\right) \right) =v(N,\{k\})
\end{eqnarray*}

\noindent since otherwise $q^{(N,\{k\})}$ would not be an optimal solution
for the coalition $(N,\{k\}).$  \hfill
\end{proof}

\bigskip One can wonder whether the set of optimal suppliers is always a
singleton. We next show with an example that the cardinality of $M^{o}$ can
be greater than one.

\begin{example}
\label{examp}Consider a MRS-situation with 1 retailer and 2 suppliers where: 
$p_{1}(q)=30-\frac{q}{10},$%
\begin{equation*}
w_{1}(q)=\left\{ 
\begin{array}{cc}
20-\frac{q}{160}, & 0\leq q\leq 1600, \\ 
10, & q>1600%
\end{array}%
\right. \text{ \ \ \ \ }c_{1}(q)=c_{2}(q)=\left\{ 
\begin{array}{cc}
15-\frac{q}{160}, & 0\leq q\leq 1200, \\ 
7.5, & q>1200.%
\end{array}%
\right.
\end{equation*}%
\begin{equation*}
w_{2}(q)=\left\{ 
\begin{array}{cc}
25-\frac{q}{110}, & 0\leq q\leq 1375, \\ 
12.5, & q>1375.%
\end{array}%
\right. \text{ \ \ \ }
\end{equation*}%
It easily follows that $q_{1}^{\ast }=300.$ Then, solving the corresponding
optimization problems for the different coalitions $(R,S)$ we obtain the
following table:

\begin{table}[!h]
\begin{center}
\begin{tabular}{|c|c|c|c|} 
 \hline
 $R$ & $S$ & $q_{1}^{(R,S)}$ & $v(R,S)$ \\ [0.5ex] 
 \hline\hline
$\{1\}$ & $\emptyset$  & $\left( 53\frac{1}{3},0\right)$  & $266\frac{2}{3}$ \\ [0.5ex] \hline
$\{1\}$ & $\{1\}$ & $\left( 80,0\right)$  & $600$ \\ [0.5ex] \hline
$\{1\}$ & $\{2\}$ & $\left( 80,0\right)$  & $600$ \\ [0.5ex] \hline
$\{1\}$ & $\{1,2\}$ & $\left( 80,0\right)$  & $600$ \\ [0.5ex] \hline
\end{tabular}
\caption{\label{table 4}}
\end{center}
\end{table}

Hence, since $v( \{1\} )=v( \{2\} )=v( \{1,2\} )$ then $M^{o}=\{1,2\}.$
\end{example}

Looking at Example \ref{examp}, $\left\vert M^{o}\right\vert =2,$
and the core of the game only has a single point: $(600,0,0)$. The following
result shows that it is true in general that if the set of optimal suppliers is not a singleton, allocations in the core always give nothing to the suppliers in the game.

\begin{theorem}
Let $(N,M,W,C,P,\infty )$ be a MRS-situation and $(N,M,v)$ be the
corresponding MRS-game with $\left\vert
M^{o}\right\vert \geq 2$. Then, for all $y\in Core(N,M,v)$, $y_{j}=0\ $for all $j\in M$.
\end{theorem}
\begin{proof}
By Theorem \ref{th:5.5} $Core(N,M,v)\neq \emptyset .$ Take a solution $%
y\in Core(N,M,v).$ Then, $y_{j}\geq 0=v(\emptyset ,\{j\})$, $\ $for all $%
j\in M$. Consider $j^{1}\in M^{o}$, we have that $v(N,j^{1})=v(N,M)=\sum_{i\in
NUM}y_{i}$ and as $\sum_{i\in NU\{j^{1}\}}y_{i}\geq v(N,j^{1})$ we have that 
$\sum_{i\in M\backslash \{j^{1}\}}y_{i}=0$. Using a similar argument with $%
j^{2}\in M^{o}$ we obtain that $\sum_{i\in M\backslash \{j^{2}\}}y_{i}=0.$
Hence, we can conclude that $y_{j}=0\ $for all $j\in M.$  \hfill
\end{proof}

\bigskip Unbounded production situations are interesting because the MRS-allocation coincides with the altruistic solution, as demonstrated by the following result.

\begin{proposition}
Let $(N,M,W,C,P,\infty )$ be a MRS-situation, $(N,M,v)$ the
corresponding MRS-game and $q^{(N,M)}$ an optimal solution. Then, $\xi(v)=x^{a}(v)$.
\end{proposition}
\begin{proof}
By Lemma \ref{lml} $\left\vert M^{o}\right\vert \geq 1$ since the MRS-situation is unbounded. Take $k\in M^{o}$, we have that $v(N,\left\{k\right\})=v(N,M)$ and  $v(N,M \setminus \left\{j\right\})=v(N,M)=\sum_{i\in N}x_{i}^{a}(v)$ by property (iii) of Proposition \ref{lemma} with $j\neq k$. Hence, $\beta=0$ and $\xi(v)=x^{a}(v)$.\hfill
\end{proof}

\bigskip We previously determined that if there are more than two optimal suppliers, any core-allocation assigns them zero profit. However, if there is only one optimal supplier, it is possible to compensate it with a small adjustment to the definition of the MRS-allocation.

\begin{definition}
Let $(N,M,W,C,P,\infty )$ be a MRS-situation, $(N,M,v)$ the corresponding
MRS-game and $q^{(N,M)}$ an optimal solution. We define the modified
SC-allocation $\xi ^{\ast }(v)=\left( \xi _{i}^{\ast }(v)\right) _{i\in
N\cup M}$ as: 
\begin{equation*}
\xi _{i}^{\ast }(v):=\left\{ 
\begin{array}{ll}
\Pi _{i}\left( q_{i}^{(N,M)},\Psi ^{M}\left( q_{N}^{(N,M)}\right) \right)
-\beta ^{\ast }, & i\in N, \\ 
\left\vert N\right\vert \cdot \beta ^{\ast } & i\in M^{o} \\ 
0, & \text{otherwise,}%
\end{array}%
\right. 
\end{equation*}%
where $\beta ^{\ast }:=\underset{\emptyset \neq R\subseteq N,j\in M^{o}}{%
\min }\left\{ \frac{\sum_{i\in R}x_{i}^{a}(v)-v(R,M\backslash \{j\})}{%
\left\vert R\right\vert }\right\} .$
\end{definition}

\bigskip The reader may notice that if $\left\vert M^{o}\right\vert \geq 2$
we have that $\beta ^{\ast }=0,$ hence $\xi ^{\ast }(v)=$ $x^{a}(v)$. In the
solution $\xi ^{\ast }(v)$, unlike the solution $\xi (v)$, each retailer compensates
with the minimal possible amount $\beta ^{\ast }$ to a single optimal
supplier. Next proposition shows that the modified SC-allocation is always a core allocation.

\begin{proposition}
Let $(N,M,W,C,P,\infty )$ be a MRS-situation, $(N,M,v)$ the corresponding
MRS-game and $q^{(N,M)}$ an optimal solution. Then, $\xi^{*} (v)\in Core(N,M,v)$.
\end{proposition}

\begin{proof}
It is trivial for $\left\vert M^{o}\right\vert \geq 2.$ Therefore, we
suppose that $\left\vert M^{o}\right\vert = 1.$ It is clear that $%
\xi^{*}(v)$ satisfies (EF). By Lemma \ref{Lemma2}, $\sum_{i\in R}\Pi
_{i}\left( q_{i}^{(N,M)},\Psi ^{M}\left( q_{N}^{(N,M)}\right) \right) \geq
v(R,S)$ for all $R\subseteq N,S\subseteq M$. Hence, if $M^{o}=\{k\}$ and $%
k\notin S:$

\begin{eqnarray*}
\sum_{i\in R\cup S}\xi _{i}^{\ast }(v) &=&\sum_{i\in
R}x_{i}^{a}(v)-\left\vert R\right\vert \cdot \beta ^{\ast }  \\
&\geq &\sum_{i\in R}x_{i}^{a}(v)-\left\vert R\right\vert \ \cdot \frac{%
\sum_{i\in R}x_{i}^{a}(v)-v(R,M\backslash \{k\})}{\left\vert R\right\vert }
\\
&=&v(R,M\backslash \{k\})\geq v(R,S).
\end{eqnarray*}

The last inequality follows by (P5) in Lemma \ref{Lemma 1}. If $k\in S$ 
\begin{eqnarray*}
\sum_{i\in R\cup S}\xi^{*} _{i}(v) &=&\sum_{i\in R}x_{i}^{a}(v)-\left\vert
R\right\vert \cdot \beta ^{\ast } +\left\vert N\right\vert \cdot \beta ^{\ast } \\
&\geq &\sum_{i\in R}x_{i}^{a}(v)\geq v(R,S).
\end{eqnarray*}%
Therefore, we conclude $\xi^{*} (v)\in Core(N,v)$. \hfill
\end{proof}

\bigskip In the following example we calculate the modified MRS-allocationusing the data in  Example \ref{examp} and compare it with the altruistic allocation.

\begin{example}

Consider again the  example \ref%
{Exep1}. Recall that the benefit of the grand coalition is $V(N,M)=2904\frac{1}{2}$. We compare the modified MRS-allocation with the altruistic allocation for this example.

\begin{center}
\begin{tabular}{c|c}
$x^{a}(v)$ & $\xi^{*} (v)$ \\ \hline
&  \\ 
$\left( 
\begin{array}{c}
777\frac{3}{11} \\ 
2127\frac{3}{11} \\ 
0 \\ 
0%
\end{array}%
\right) $ & $\left( 
\begin{array}{c}
275 \\ 
1625 \\ 
1004\frac{6}{11} \\ 
0%
\end{array}%
\right) $%
\end{tabular}%
\end{center}

with $\beta :=\underset{\emptyset \neq R\subseteq N}{\min }\left\{ \frac{%
\sum_{i\in R}x_{i}^{a}(v)-v(R,\{1\})}{\left\vert R\right\vert\ }\right\} =\min \{777\frac{3}{11}%
-275,2127\frac{3}{11}-770,\frac{2904\frac{1}{2}-1342\frac{1}{3}}{2}\}=502%
\frac{3}{11}.$

The reader may notive that, as we have already pointed out, the modified SC-allocation reduces the retailers' profit $(275$,$1625)$ to give a profit to the cheapest supplier $(1004\frac{6}{11})$.

\end{example}

\section{Concluding remarks}

The paper Guardiola et al. (2007) studied cooperation and profit allocation in single-period supply chains   composed by one supplier and several non-competing retailers. In this paper, we further extend the model of Guardiola et al. (2007) to the case of multiple suppliers.  We consider single-period distribution chains with a single product. In this supply chain, retailers place their orders at the suppliers one-time. Each retailer chooses its order quantity to maximize its profit. We assume that each retailer pays each supplier a wholesale price per unit product that decreases with the quantity ordered.  Therefore, there exist incentives for cooperation among retailers. Additionally, because of the absence of the intermediate wholesale prices, retailers are willing to cooperate with suppliers in order to obtain a further reduction in cost inefficiency. Clearly, the total profit under full cooperation is larger than the sum of the individual profits.

We define the multiretailer-multisupplier cooperative game (MRS-game) where the value of a coalition of retailers and a coalition of suppliers equal to the optimal profit they can
achieve together. We prove that the value of any coalition of retailers and suppliers is positive and that each MRS-game is superadditive. Consequently, retailers and suppliers have incentives to form the grand coalition in MRS games. We show that MRS-games are balanced. Further, we propose a stable allocation of the profits induced by full cooperation (a core allocation), the so called SC-allocation. This is an appealing core allocation that takes into account the importance of the  suppliers to achieve full cooperation. In addition, we provide a characterization of the SC-allocation. Finally, we focus on the MRS-games where the suppliers can produce sufficiently large quantities of product so that retailers are not constrained in their orders' sizes. We provide a core allocation that allows retailers to compensate a single optimal supplier.

Further research can look at alternative modes of relationships among the agents Particularly, it would be interesting to analyse situations  where a group of retailers can enter and sell directly their orders in other retailers markets, provided that they obtain a higher selling price, even if this policy may possibly cause undersupply their own markets.

\bigskip

\bigskip


\begin{thebibliography}{99}

\bibitem{B01} Ball, M. A. A new solution for n-person games using coalitional theory. I. The conditions. \textit{Proceedings of the Royal Society of London. Series A: Mathematical, Physical and Engineering Sciences} \textbf{2001}, 457(2005), 95-116.

\bibitem{B63} Bondareva, O.N. Some applications of linear programming
methods to the theory of cooperative games. \textit{Problemy Kibernety} 
\textbf{1963}, 10, 119-139. In Russian.

\bibitem{CWZ19} Chen, X., Wang, X., \& Zhou, M. Firms’ green R\&D cooperation behaviour in a supply chain: Technological spillover, power and coordination. \textit{International Journal of Production Economics}  \textbf{2019}, 218, 118-134.

\bibitem{DL17} Deng, A., \& Li, Y. A Review of Research on Interest Distribution of Supply Chain Based on Game Theory. \textit{World Journal of Research and Review} \textbf{2017}, 4(6), 262789.

\bibitem{DK97} Drexl, A., \& Kimms, A. Lot sizing and scheduling—survey and extensions. \textit{European Journal of operational research} \textbf{1997}, 99(2), 221-235.

\bibitem{E15} Elomri, A. Cooperation in supply chain networks: Motives, outcomes, and barriers. \textit{Int J Sup Chain Mgt} \textbf{2015}, 4(1), 10-1081.

\bibitem{IB97} Iyer, A. V., \& Bergen, M. E. Quick response in manufacturer-retailer channels. \textit{Management science} \textbf{1997}, 43(4), 559-570.

\bibitem{GDGJFJ10} Gonz\'alez-D\'{\i}az, J., Garc\'{\i}a-Jurado, I. and Fiestras-Janeiro, M.G.  An introductory course on mathematical game theory. Graduate Studies in Mathematics 115. American Mathematical Society \textbf{2010}.	

\bibitem{GKT99} Gavirneni, S., Kapuscinski, R., \& Tayur, S. Value of information in capacitated supply chains. \textit{Management science} \textbf{1999}, 45(1), 16-24

\bibitem{GY04} Granot, D., \& Yin, S. Competition and cooperation in a multi-manufacturer single-retailer supply chain with complementary products. \textit{working paper} \textbf{2004}, Sauder School of Business, University of British Columbia, Vancouver, Canada.

\bibitem{G07} Guardiola, L. A., Meca, A., \& Timmer, J. Cooperation and
profit allocation in distribution chains. \textit{Decision support systems} 
\textbf{2007}, 44(1), 17-27.

\bibitem{HHS21} Halat, K., Hafezalkotob, A., \& Sayadi, M. K. Cooperative inventory games in multi-echelon supply chains under carbon tax policy: Vertical or horizontal?. \textit{Applied Mathematical Modelling} \textbf{2021}, 99, 166-203.

\bibitem{HM89} Hart, S., \& Mas-Colell, A. Potential, value, and consistency. \textit{Econometrica: Journal of the Econometric Society} \textbf{1989}, 589-614.

\bibitem{HDS00} Hartman, B. C., Dror, M., \& Shaked, M. Cores of inventory centralization games. \textit{Games and Economic Behavior} \textbf{2000}, 31(1), 26-49.

\bibitem{LP05} Leng, M., \& Parlar, M. Game theoretic applications in supply chain management: a review. \textit{INFOR: Information Systems and Operational Research} \textbf{2005}, 43(3), 187-220.

\bibitem{MGA09} Marcotte, F., Grabot, B., \& Affonso, R. Cooperation models for supply chain management. \textit{International Journal of Logistics Systems and Management} \textbf{2009}, 5(1-2), 123-153.

\bibitem{MT08} Meca, A., \& Timmer, J. Supply chain collaboration. \textbf{2008}, Supply Chain, 1.

\bibitem{MGH14} Ming, Y., Grabot, B., \& Houé, R. A typology of the situations of cooperation in supply chains. \textit{Computers \& Industrial Engineering} \textbf{2014}, 67, 56-71.

\bibitem{MABJ18} Moradinasab, N., Amin-Naseri, M. R., Behbahani, T. J., \& Jafarzadeh, H. Competition and cooperation between supply chains in multi-objective petroleum green supply chain: A game theoretic approach. \textit{Journal of Cleaner Production} \textbf{2018}, 170, 818-841.

\bibitem{MSS02} Müller, A., Scarsini, M., \& Shaked, M. The newsvendor game has a nonempty core. \textit{Games and Economic Behavior} \textbf{2002}, 38(1), 118-126.

\bibitem{NS08} Nagarajan, M., \& Sošić, G. Game-theoretic analysis of cooperation among supply chain agents: Review and extensions. \textit{European journal of operational research} \textbf{2008}, 187(3), 719-745.

\bibitem{NS09} Nagarajan, M., \& Sošić, G. Coalition stability in assembly models.  \textit{Operations Research} \textbf{2009}, 57(1), 131-145.

\bibitem{PPF09} Perea, F., Puerto, J., \& Fernández, F. R. (2009). Modeling cooperation on a class of distribution problems.  \textit{European Journal of Operational Research} \textbf{2009}, 198(3), 726-733.

\bibitem{SH67} Shapley, L.S. On Balanced Sets and Cores. \textit{Naval Research Logistics} \textbf{1967}, 14, 453-460.

\bibitem{SFW05} Slikker, M., Fransoo, J., \& Wouters, M. Cooperation between multiple news-vendors with transshipments. \textit{European Journal of Operational Research} \textbf{2005}, 167(2), 370-380.

\bibitem{TBT04} Timmer, J., Borm, P., \& Tijs, S. On three Shapley-like solutions for cooperative games with random payoffs. \textit{International Journal of Game Theory} \textbf{2004}, 32(4), 595-613.

\end{thebibliography}
\end{document}